\documentclass{article}
\usepackage{spconf,amsmath,graphicx}

\usepackage{mathtools}
\usepackage{amsmath}
\usepackage{algorithm,algpseudocode}
\usepackage{hyperref}
\usepackage{amsfonts,bm, dsfont, amssymb, amsthm, verbatim}
\usepackage{enumitem}
\usepackage{graphicx}
\usepackage[font={small,it}]{caption}
\usepackage{tikz}
\usepackage{multicol}
\usepackage{subfig}

\algnewcommand{\Inputs}[1]{%
	\State \textbf{Inputs:}
	\Statex \hspace*{\algorithmicindent}\parbox[t]{.8\linewidth}{\raggedright #1}
}
\algnewcommand{\Initialize}[1]{%
	\State \textbf{Initialize:}
	\Statex \hspace*{\algorithmicindent}\parbox[t]{.8\linewidth}{\raggedright #1}
}

\newcommand{\ma}[1]{\ensuremath{\mathsf{#1}}}
\renewcommand{\vec}[1]{\ensuremath{\mathbf{#1}}}

\renewcommand{\O}{\ensuremath{\mathcal{O}}}
\newcommand{\R}{\ensuremath{\mathbb{R}}}
\newcommand{\E}{\ensuremath{\mathbb{E}}}

\newtheorem{proposition}{Proposition}

\newcommand{\ie}{\textit{i.e.}}
\newcommand{\eg}{\textit{e.g.}}
\newcommand{\etal}{\textit{et al.}}

\DeclareMathOperator*{\argmin}{\mathrm{argmin}}

\DeclareMathOperator*{\Var}{Var}
\DeclareMathOperator*{\Cov}{Cov}

\DeclareMathOperator*{\tr}{tr}

\algdef{SE}[DOWHILE]{Do}{doWhile}{\algorithmicdo}[1]{\algorithmicwhile\ #1}%
\def\x{{\mathbf x}}

\title{VARIANCE REDUCTION IN STOCHASTIC METHODS FOR LARGE-SCALE REGULARISED LEAST-SQUARES PROBLEMS}
\name{Yusuf Yi\u{g}it Pilavc\i, Pierre-Olivier Amblard, Simon Barthelm\'e, Nicolas Tremblay\thanks{This work was partly funded by
		the French National Research Agency in the framework of the "Investissements d’avenir” program (ANR-15-IDEX-02), the LabEx PERSYVAL (ANR-11-LABX-0025-01), the ANR GraVa (ANR-18-CE40-0005),
		the ANR GRANOLA  (ANR-21-CE48-0009),
		the MIAI@Grenoble Alpes chairs ``LargeDATA at UGA" and ``Pollutants'' (ANR-19-P3IA-0003).}}
\address{	CNRS, Univ. Grenoble Alpes, Grenoble INP, GIPSA-lab, Grenoble, France}
\begin{document}
%
\maketitle
\begin{abstract}
  Large dimensional least-squares and regularised least-squares problems are
  expensive to solve. There exist many approximate techniques, some
  deterministic (like conjugate gradient), some stochastic (like stochastic
  gradient descent). Among the latter, a new class of techniques uses
  Determinantal Point Processes (DPPs) to produce unbiased estimators of the
  solution. In particular, they can be used to perform Tikhonov regularization
  on graphs using random spanning forests, a specific DPP.
  While the unbiasedness of these algorithms is attractive, their variance can
  be high. We show here that variance can be reduced by combining the stochastic
  estimator with a deterministic gradient-descent step, while keeping the
  property of unbiasedness. We apply this technique to Tikhonov regularization
  on graphs, where the reduction in variance is found to be substantial at very
  small extra cost.


\end{abstract}
\section{Introduction}
\label{sec:intro}

In linear least-squares problems, given measurements $\vec{y} \in \R^n$ and
predictors $\ma{A} \in \R^{n \times p}$, we seek the vector $\hat{\vec{x}}$ verifying:
\begin{equation}
  \label{eq:least-squares}
  \hat{\vec{x}}=\argmin_{\vec{x}\in\mathbb{R}^p} ||\ma{A}\vec{x}-\vec{y}||_2^2 + \lambda
  \vec{x}^t\ma{P}\vec{x}
\end{equation}
where $\lambda \vec{x}^t\ma{P}\vec{x}$ is a regularization term.
$\hat{\vec{x}}$ can be computed exactly at cost $\O(np^2)$, but for very large $n$
 or $p$ this is costly, and approximate methods may be used instead. There are
 dozens of such methods, either deterministic or stochastic in nature. The best
 known example of the former are the various gradient descent methods \cite{nesterov2003introductory},
 while among the latter the most popular is certainly stochastic gradient
 descent \cite{bottou-mlss-2004}. An attractive alternative to stochastic gradient descent is to
 use methods based on Determinantal Point Processes, which, roughly, consist in
 finding a  well-chosen, random subset of the rows of $\ma{A}$ so that solving
 the least-squares problem for \emph{just these rows} gives an \emph{unbiased}
 estimator for the solution to the full problem \cite{derezinski2021determinantal,fanuel2021diversity}. This contrasts with
 stochastic gradient descent, where the estimator is generally biased after a
 finite number of steps.

 One drawback of this new class of methods is that sampling a subset of rows
 with the right properties may be expensive, so that reducing the variance by simply increasing the number of estimates becomes rapidly too costly. 
 We show below that a very
 simple way to reduce variance is to combine a stochastic estimator with (at
 least one) step of gradient descent. The result is another unbiased estimator,
 one that is guaranteed to have lower variance. We apply our technique to
 Tikhonov regularization on graphs, where DPP-based estimators are particularly
 well-adapted \cite{pilavci2020smoothing,pilavci2021graph}. As in standard applications of gradient descent, most of
 the difficulty consists in determining how large the gradient descent step
 should be. In graphs this problem turns out to be quite tractable, with good
 heuristics available. Numerical results confirm that the reduction in variance
 obtained is substantial, at small computational cost.\\

\noindent \textbf{Main idea} The main idea of the paper is extremely simple. Solving a least-squares problem of the form of Eq.\eqref{eq:least-squares} is equivalent to
minimizing a quadratic form
\begin{equation}
  \label{eq:quad-form}
  f(\x) = \frac{1}{2} \x^t \ma{Q} \x - \vec{r}^t \x.
\end{equation}
Minimising $f$ by gradient descent consists in taking steps of the form
\[ \x_t = \x_{t-1} - \alpha \nabla f(\x_{t-1})  = \x_{t-1} - \alpha
  (\ma{Q}\x_{t-1} - \vec{r}),\] where $\alpha$ is the step size.
Since the gradient is zero at the solution, $\hat{\x} = \argmin
f(\x)$ is a fixed point of the iteration. Now suppose that $\bar{\x}$ is an
unbiased estimator of $\hat{\x}$, \ie, a random variable such that $\E
(\bar{\x}) = \hat{\x} $. Then, it is easy to check that
\[ \E \left(\bar{\x}  - \alpha (\ma{Q}\bar{\x} - \vec{r}) \right) = \hat{\x}, \]
meaning that $\bar{\x}$ stays unbiased after one step of gradient descent (or
indeed several). We show below that setting the step size $\alpha$ correctly guarantees a
reduction in variance, and give practical solutions for finding an appropriate
step size in a graph signal processing setting.




\section{Background}
\label{sec:background}

\noindent \textbf{Tikhonov regularization in graphs.}
The regularised least-squares estimator we are interested in is graph Tikhonov
regularisation (GTR), a method for denoising
graph signals. A ``graph signal'' is a vector of measurements associated with
the nodes of a graph, for instance brain activity in $n$ brain regions, where
the graph models neural connectivity across regions. GTR also occurs as a
subproblem in other methods, like semi-supervised learning
\cite{pilavci2021graph}, which is why finding an efficient approximation
algorithm is of high interest.

Let us first set notation. We denote a graph by the set $\mathcal{G} = (\mathcal{V},\mathcal{E},w)$ with  $|\mathcal{V}| = n$ vertices and $|\mathcal{E}| = m$ edges. $w:\mathcal{V}\times\mathcal{V}\mapsto\mathbb{R}^{+}$ is called the weight function and is non-zero only for $(i,j)\in\mathcal{E}$. In this work, we only consider connected (\ie, there exists a path between any pair of nodes) and undirected graphs (\ie, verifying $\forall(i,j)\in\mathcal{E}, w(i,j) =w(j,i)$).
In addition, we use certain matrices to depict the algebraic properties of graphs.
These are the (weighted) adjacency matrix $\ma{W} = [w(i,j)]_{i,j}\in\mathbb{R}^{n\times n}$, the diagonal degree matrix whose diagonal entries are $\ma{D}_{i,i} = \sum_{j=1}^n w(i,j)$ and the graph Laplacian, $\ma{L} = \ma{D}-\ma{W}$. It is well-known that $\ma{L}\in\mathbb{R}^{n\times n}$ is a symmetric positive semi-definite matrix with $n$ eigenvalues $\lambda_1=0<\lambda_2\leq\hdots\leq\lambda_n$ for any undirected, connected  graph~\cite{van2010graph}. 
\smallskip

Given a noisy version
$\vec{y}\in\mathbb{R}^n$ of an underlying graph signal $\vec{x}\in\mathbb{R}^n$ that one wishes to recover, GTR consists in solving the following minimisation problem to estimate $\vec{x}$:
\begin{equation}
	\hat{\vec{x}}= \argmin_{\vec{z}\in\mathbb{R}^n} q||\vec{z}-\vec{y}||_2^2 + \vec{z}^\top\ma{L}\vec{z},
	\label{eq:tikhonov}
\end{equation}
where
$q\in\mathbb{R}^+$ is a parameter that adjusts the balance between the data
fidelity term $||\vec{z}-\vec{y}||_2^2$ and the regularization term
$\vec{z}^\top\ma{L}\vec{z}$, which forces the solution to be smooth over the
graph. Note that this is a special case of eq. \eqref{eq:least-squares} with
$\ma{A} = \ma{I}$, $\ma{P} = \ma{L}$ and $\lambda=1/q$.
The explicit solution to this problem reads:
\[
	\hat{\vec{x}} = \ma{K}\vec{y} ~\text{ with }~ \ma{K} = q(q\ma{I} +\ma{L} )^{-1}
\]
The direct computation of $\hat{\vec{x}}$ requires (a worst-case) $\mathcal{O}(n^3)$ elementary operations due to the inversion of $(q\ma{I} +\ma{L})$. As $n$ increases, this computation becomes prohibitive. In the next section, we list state-of-the-art methods which avoid this expensive computation.\\

\noindent \textbf{DPP-based estimators for Graph Tikhonov Regularisation.}
In the special case of GTR, the most popular methods are deterministic.
For large $n$, state-of-the-art methods can be divided roughly into two groups, namely iterative methods (\eg~conjugate gradient~\cite{saad2003iterative}) and polynomial approximations (\eg~Chebyshev polynomials~\cite{shuman2011chebyshev}).
Both classes of methods run in linear time with the number of edges, $m$.

As an alternative stochastic method for GTR, we have proposed in previous works unbiased Monte
Carlo estimators for approximating $\hat{\vec{x}}$ which also scale
linearly with $m$ ~\cite{pilavci2020smoothing,pilavci2021graph}. Our method can
be viewed as a DPP-based subsampling of the rows of a matrix, but the case of
GTR is particularly favourable. The appropriate DPP is the random forest
process \cite{avena2018two}, which is easy to sample from, and its analytic properties enable
variance reduction via conditional Monte Carlo. The resulting estimator,
called $\bar{\x}$ in ~\cite{pilavci2020smoothing,pilavci2021graph}, is unbiased ($\mathbb{E}(\bar{\vec{x}})=\hat{\vec{x}}$), can be obtained in $\mathcal{O}(m)$ elementary operations, and takes a very
simple form: based on the partition formed by a random spanning forest, it averages the signal in each part of the partition. We
repeat this process $N$ times and average to get a denoised signal. The details
of the estimator can be found in ~\cite{pilavci2021graph}, but are irrelevant
for what we discuss here.



\section{Proposed Method}
\label{sec:alpha_star}
In this section, we propose an improved estimator for estimating
$\ma{K}\vec{y}$. As described in the introduction, we combine our unbiased
estimator with a gradient descent step to reduce variance. This variance
reduction technique can also be viewed as an instance of the control variate
method~\cite{botev2017variance} in the Monte Carlo literature.

\subsection{The gradient descent update}
Note that the solution to \eqref{eq:tikhonov} also minimizes the cost function:
\[
F({\vec{z}}) = \frac{1}{2}\vec{z}^\top\ma{K}^{-1}\vec{z} - \vec{z}^\top\vec{y}.
\]
As stated in the introduction, applying the gradient step to $\bar{\vec{x}}$ yields a new estimator:
\begin{equation}
\label{eq:defzbar}
	\bar{\vec{z}} \coloneqq \bar{\vec{x}} - \alpha ( \ma{K}^{-1}\bar{\vec{x}} - \vec{y}).
\end{equation}
\begin{proposition}
	$\bar{\vec{z}}$ is an unbiased estimator for $\hat{\vec{x}}$. Moreover, assuming\footnote{Looking closely at the definition of $\bar{x}$ in~\cite{pilavci2021graph} as well as its variance property, one can show that $\tr\left(\Var(\ma{K}^{-1}\bar{\vec{x}})\right)=0$ only arises when $\vec{y}$ is a constant vector, in which case $\bar{\vec{x}}$ is always equal to $\vec{y}$ whatever the sampled forest, such that by Eq.~\eqref{eq:defzbar}, $\bar{\vec{z}}=\bar{\vec{x}}$ and there is no point in using $\bar{\vec{z}}$.} $\tr\left(\Var(\ma{K}^{-1}\bar{\vec{x}})\right)>0$, the MSE of $\bar{\vec{z}}$ is a quadratic function of $\alpha$ which is minimized for:
		\begin{equation}
		\label{eq:alphastar_covovervar}
			\alpha^\star = \frac{\tr\left(\Cov(\ma{K}^{-1}\bar{\vec{x}},\bar{\vec{x}})\right)}{\tr\left(\Var(\ma{K}^{-1}\bar{\vec{x}})\right)}.
		\end{equation}
    \end{proposition}
	\begin{proof}
		Since $\bar{\vec{x}}$ is an unbiased estimator, the expectation of $\bar{\vec{z}}$ immediately reads
		$\mathbb{E}[\bar{\vec{z}}]= \mathbb{E}[\bar{\vec{x}}] - \alpha(\ma{K}^{-1}\mathbb{E}[\bar{\vec{x}}] - \vec{y}) = \hat{\vec{x}}.$
		Now, let us focus on the variance of this new estimator:
		\begin{equation}
		\begin{split}
				\tr\left(\Var(\bar{\vec{z}})\right) =
				\tr\left(\Var(\bar{\vec{x}})\right) &+ \alpha^2\tr\left(\Var(\ma{K}^{-1}\bar{\vec{x}})\right) \\
				&-2\alpha\tr\left(\Cov(\ma{K}^{-1}\bar{\vec{x}},\bar{\vec{x}})\right).
			\end{split}
		\label{eq:varianceparabola}
		\end{equation}
		 Eq.~\eqref{eq:varianceparabola} is a quadratic function in $\alpha$ and is minimized at $\alpha^\star$ given in Eq.~\eqref{eq:alphastar_covovervar}:
	\end{proof}

\noindent \textbf{Connection with control variates.}
The proposed method is a particular instance of the control variate method.
This technique leverages an additional random variable, called control variate, with a known expectation to improve the Monte Carlo performance.
In our case, the control variate is $\ma{K}^{-1}\bar{\vec{x}}$ whose expectation is the input signal $\vec{y}$.
\smallskip

\noindent \textbf{Implementation.}
$\bar{\vec{z}}$  requires the computation of the control variate per sample in addition to $\bar{\vec{x}}$.
The nontrivial part of this computation includes the matrix-vector product $\ma{L}\bar{\vec{x}}$.
However, we can avoid to repeat this product per sample.
Instead, we can calculate the sample mean $\frac{1}{N}\sum_{i=1}^N\bar{\vec{x}}^{(i)}$, then do the product once to directly compute the sample mean for $\bar{\vec{z}}$.

\subsection{How to select $\alpha$?}
Unfortunately, calculating the optimal value $\alpha^\star$ requires information that is not readily available.  However, as the MSE (Eq.~\eqref{eq:varianceparabola}) is a quadratic function in $\alpha$ minimized in $\alpha^\star$, we know that any choice of $\alpha\in(0,2\alpha^\star)$ will necessarily decrease the variance (see an illustration of this in Fig.~\ref{fig:reg_vs_irreg_alpha}). We now discuss several options to find such appropriate --yet not optimal-- $\alpha$.
Our examination can be separated into two parts, each proposing a choice of $\alpha$ with different motivations.
First, we look for a constant for $\alpha$ that is cheap to compute and ensures variance reduction.
Second, we consider how to approximate $\alpha^\star$ within the Monte Carlo simulations.
\medskip

\noindent \textbf{A first option}. To ensure variance reduction, $\bar{\vec{z}}$ must verify:
\[
\mathbb{E}[|| \bar{\vec{z}}- \hat{\vec{x}} ||_2] \leq \mathbb{E}[|| \bar{\vec{x}} - \hat{\vec{x}}||_2].
\]
Notice that this inequality is satisfied if the following is true for all $\bar{\vec{x}},\bar{\vec{z}}\in\mathbb{R}^n$:
\begin{equation}
|| \bar{\vec{z}}- \hat{\vec{x}} ||_2 \leq || \bar{\vec{x}} - \hat{\vec{x}}||_2.
\label{eq:zbarineq}
\end{equation}
We use this inequality to find a safe range for $\alpha$ and propose to choose the largest $\alpha$ within this range. This particular choice costs $\mathcal{O}(1)$ to calculate for a given graph.
\begin{proposition}
\label{prop:saferange}
Let $d_{max}$ be the maximum degree of the graph. Then, setting $\alpha = \frac{2q}{q+2d_{max}}$ ensures that:
\[
|| \bar{\vec{z}}- \hat{\vec{x}} ||_2 \leq || \bar{\vec{x}} - \hat{\vec{x}}||_2.
\]
\begin{proof}
The following must hold to satisfy Eq.~\eqref{eq:zbarineq}
\begin{equation}
	\begin{split}
	|| \bar{\vec{x}} - \hat{\vec{x}} -\alpha(\ma{K}^{-1}\bar{\vec{x}} - \vec{y}) ||_2 &\leq || \bar{\vec{x}} - \hat{\vec{x}}||_2 \\
	||  (\ma{I} - \alpha\ma{K}^{-1})(\bar{\vec{x}} - \hat{\vec{x}}) ||_2 &\leq || \bar{\vec{x}} - \hat{\vec{x}}||_2. \\
	\end{split}
\end{equation}
This inequality holds for all $\bar{\vec{x}},\hat{\vec{x}}$ if and only if $(\ma{I} - \alpha\ma{K}^{-1})$ is a contraction mapping which means:
\[	|\mu_i| \leq 1, \quad \forall i\in \{1,\hdots,n\},\]
where $\mu_i$'s are the eigenvalues of the matrix $(\ma{I} - \alpha\ma{K}^{-1})$. Rewriting this constraint in terms of the eigenvalues of the graph Laplacian, we have: $\forall i,~ |1-\alpha\frac{\lambda_i}{q}-\alpha| \leq 1$, \ie: $\forall i\quad	0\leq\alpha \leq \frac{2q}{ \lambda_{i} +q}$.
The tightest upper bound is given by $\lambda_i=\lambda_n$ and recalling $\lambda_n\leq2d_{max}$~\cite{van2010graph} finishes the proof.
\end{proof}
\end{proposition}
\noindent \textbf{A second option}. In this section, we approximate $\alpha^\star$ from the samples.
Let us define $\bar{\vec{y}} \coloneqq \ma{K}^{-1}\vec{\bar{x}}$. Then, Eq.~\eqref{eq:alphastar_covovervar} may be re-written as 
$
\alpha^\star= \frac{
	\tr(\Cov(\vec{\bar{x}},\vec{\bar{y}}))}{\tr(\Var(\vec{\bar{y}}))}.$ 
Using $N$ samples $\bar{\vec{x}}^{(1)},\dots,\bar{\vec{x}}^{(N)}$ and the  corresponding $\bar{\vec{y}}^{(1)},\dots,\bar{\vec{y}}^{(N)}$, we can empirically estimate $\alpha^*$ with the sample covariance matrices:
\begin{equation}
\hat{\alpha}= \frac{
	\tr(\widehat{\Cov}(\vec{\bar{x}},\vec{\bar{y}}))}{\tr(\widehat{\Var}(\vec{\bar{y}}))}. 	\label{eq:alphahat}
\end{equation}
 The analysis in~\cite[Ch~8.9]{mcbook} shows that along with
$\hat{\alpha}$, $\bar{\vec{z}}$ may yield a biased estimator. However, there are heuristics indicating that the bias
diminishes with rate $\mathcal{O}(N^{-1})$ which is often negligible for
large $N$ w.r.t. the estimation error.\medskip
\begin{figure}
	\includegraphics[width=9cm,height=4cm]{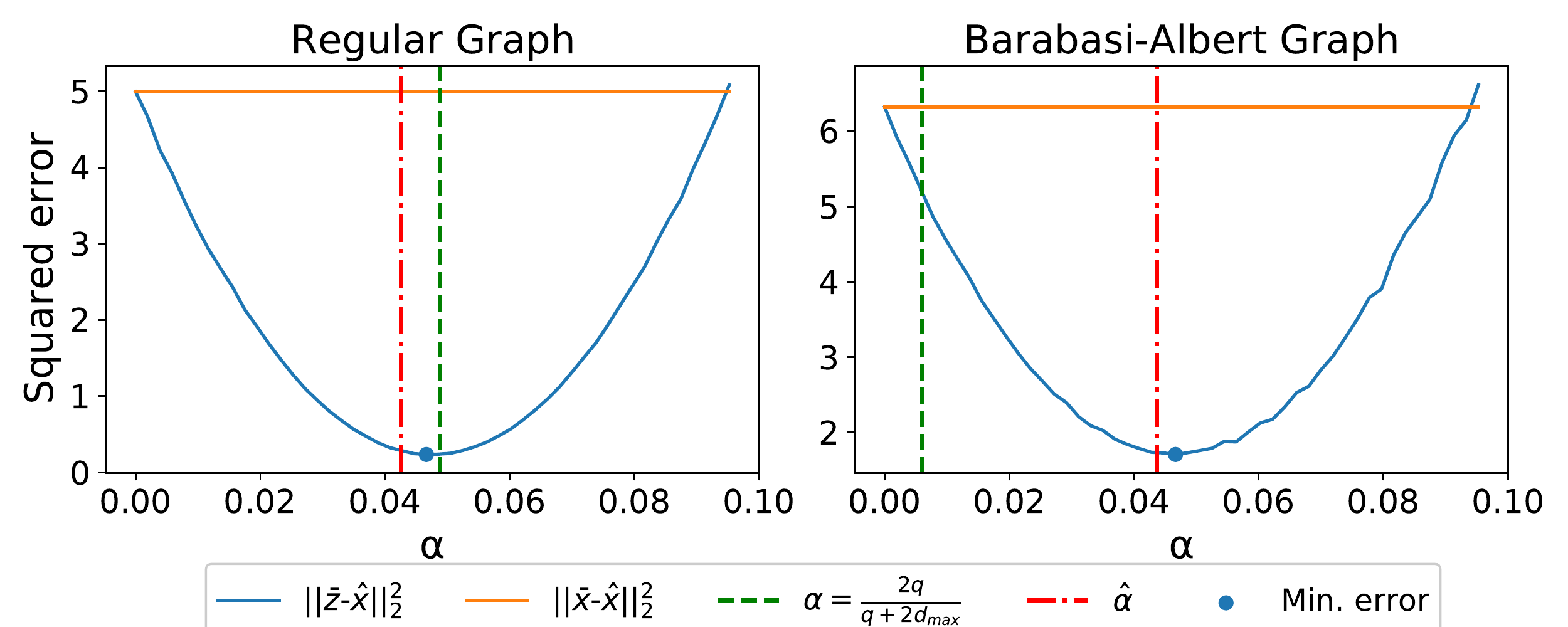}
	\caption{Empirical squared error of $\bar{\vec{x}}$ (orange horizontal line) and $\bar{\vec{z}}$ (blue parabola) w.r.t $\alpha$ on two graphs generated by random models. On the left is a random regular graph with $n=1000$ and $m=10000$. On the right is a Barabasi-Albert model with parameter $k=10$ resulting in $n=1000$ and $m=9900$.
	The green (resp. red) vertical dashed line shows $\alpha=\frac{2q}{q+2d_{max}}$ (resp. the estimated $\hat{\alpha}$ from the samples). The blue dot represents the best possible variance reduction obtained for $\alpha=\alpha^\star$. The number of Monte Carlo samples is set to $N=10$ and the error results are averaged over 200 realizations. 
	The signal is a random vector generated from $\mathcal{N}(\vec{0},\ma{I})$. }
	\label{fig:reg_vs_irreg_alpha}
\end{figure}

\noindent\textbf{An empirical comparison} of these two options for $\alpha$ is given in Fig.~\ref{fig:reg_vs_irreg_alpha}.
In a graph with a regular degree distribution, fixing $\alpha$ to $\frac{2q}{q+2d_{max}}$ gives a performance close to optimum. 
However, in the case of a graph with a broad degree distribution, 
$\frac{2q}{q+2d_{max}}$ fails to provide a good approximation for $\alpha^\star$, and one needs to resort to $\hat{\alpha}$ in order to leverage the full variance reduction potential offered by this gradient descent step.

\section{Experiments}
\label{sec:exp}
We illustrate the improved estimator $\bar{\vec{z}}$ over real data sets and compare it with $\bar{\vec{x}}$.
\smallskip

\noindent \textbf{Tikhonov denoising.} The first data set consists of the trip records of taxis in New York City~\footnote{The dataset is available at \url{https://www1.nyc.gov/site/tlc/about/tlc-trip-record-data.page}} over 260 zones. From their geographical positions, we build a $k=5$-nearest neighbour graph. The signal, in our case, is the median of the total amount charged to passengers at 260 drop-off zones. In the experiments, we add artificial Gaussian noise to the signal and reconstruct it via Tikhonov denoising and its estimates $\bar{\vec{x}}$ and $\bar{\vec{z}}$. Fig.~\ref{fig:denoising} summarizes the results.
\begin{figure}[t!]
	\centering
	\subfloat[Original Signal]{\includegraphics[width=4cm]{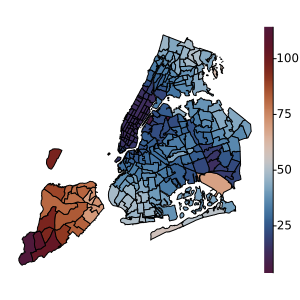}\label{fig:denoising:original}}
	\subfloat[Noisy Measurements $\vec{y}$ ]{\includegraphics[width=4cm]{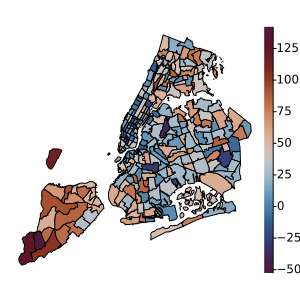}} \\
	\subfloat[Exact solution $\hat{\vec{x}}$ ]{\includegraphics[width=4cm]{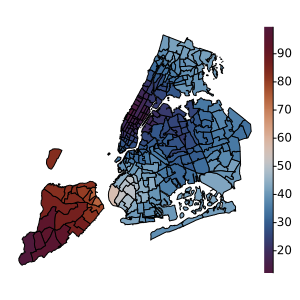}}
	\subfloat[PSNR vs q, $N=2$]{\includegraphics[width=4cm,height=4.25cm]{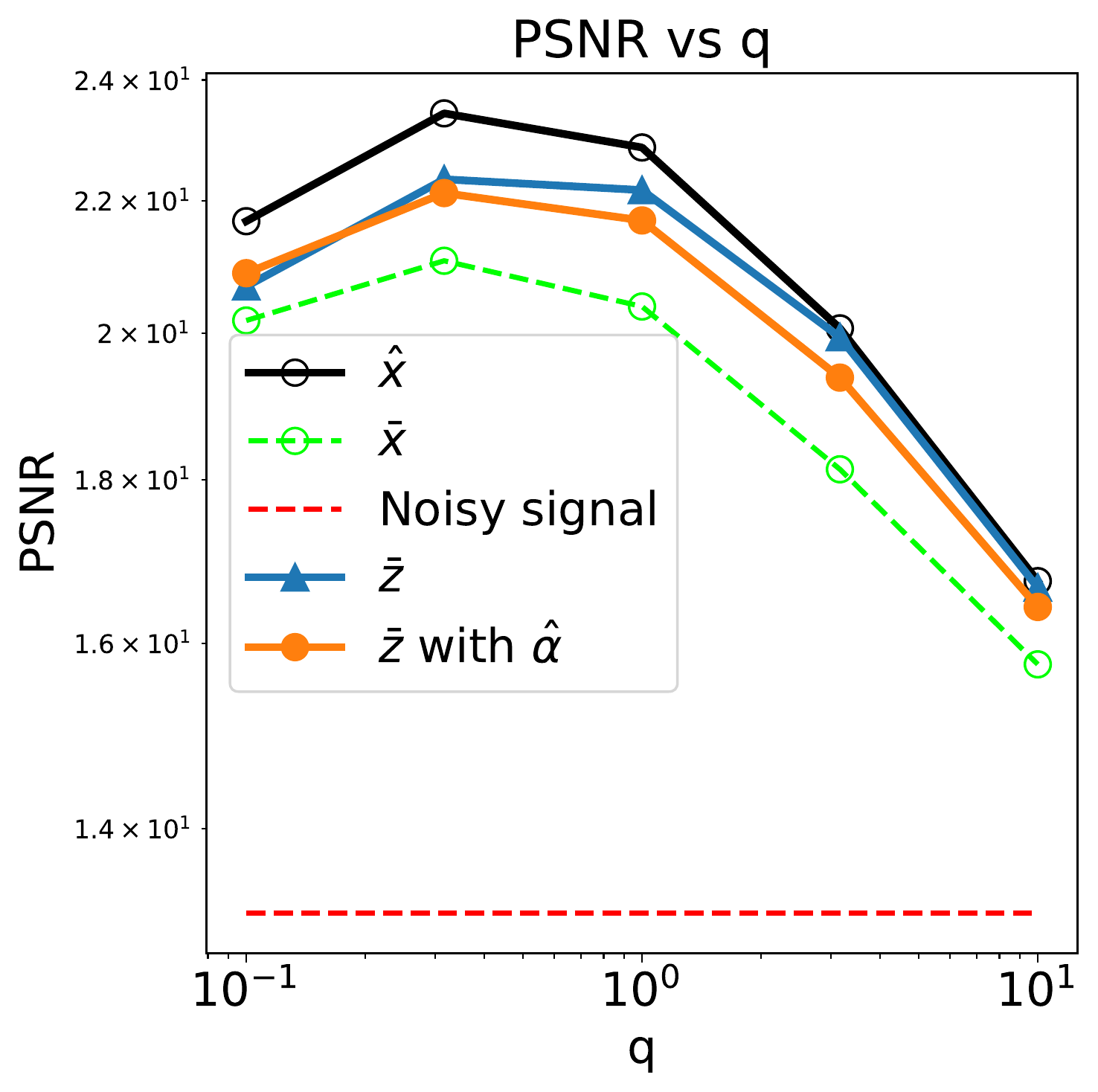}\label{fig:denoising:psnr}}  \\
	\caption{Tikhonov denoising on New York City's taxi network. A $k=5$-nn graph with $|\mathcal{V}|=260,|\mathcal{E}|=777$ is built based on the proximity of the zones. The original signal in (a) is the median of the total amount charged to passengers in drop-off locations. The noisy signal (b) is generated by adding noise drawn from a Gaussian distribution $\mathcal{N}(\vec{x},25)$. (c) is the exact solution  $\ma{K}\vec{y}$ with $q=0.32$. In (d), peak signal to noise ratio is plotted for $\vec{y},\bar{\vec{x}},\bar{\vec{z}}$ with $\alpha=\frac{2q}{q + 2d_{max}}$ and $\bar{\vec{z}}$ with $\hat{\alpha}$ as in Eq.~\eqref{eq:alphahat}.
	 }
	\label{fig:denoising}
\end{figure}
In Fig.~\ref{fig:denoising:psnr}, we can clearly see the improvement via $\bar{\vec{z}}$ w.r.t. $\bar{\vec{x}}$ in denoising performance. We observe that $\bar{\vec{z}}$ with a constant $\alpha$ slightly outperforms the case with $\hat{\alpha}$. This is probably due to the poor estimation of $\hat{\alpha}$ with $N=2$.
\medskip

\noindent \textbf{Node classification.} In the second illustration, we use $\bar{\vec{z}}$ for solving the semi-supervised node classification problem.
The purpose is to classify each vertex by leveraging the underlying graph while the class information is available over only a few vertices.
We define a labeling $\ma{Y} = [\vec{y}_1,\hdots,\vec{y}_k]\in\mathbb{R}^{n\times k}$ for $k$ classes where $\vec{y}_l(i)$ is $1$ whenever node $i$ is known to be in class $j$, otherwise it is $0$.
Given the set of labeled vertices  $\ell$, typically $|\ell|\ll|\mathcal{V}|$, the goal is to find a classification function $\ma{F}=[\vec{f}_1,\hdots,\vec{f}_k]\in\mathbb{R}^{n\times k}$.
One approach is generalized semi supervised learning framework by Avrachevkov~\etal~\cite{avrachenkov2012generalized}   which calculates the following solution:
\begin{equation}
	\ma{F} = \ma{D}^{1-\sigma}\ma{K}\ma{D}^{\sigma-1}\ma{Y}  \quad \text{ with } \ma{K} = (\ma{D} + \frac{2}{\mu}\ma{L})^{-1}\ma{D}
	\label{eq:gSSLsolution}
\end{equation}
where $\mu\in\mathbb{R}^{+}$ and $\sigma\in[0,1]$ are the hyperparameters of the algorithm. The cumbersome operation in this solution occurs in the calculations involving $\ma{K}$ due to the matrix inversion.
In~\cite{pilavci2020smoothing}, we have already showed that $\bar{\vec{x}}$ can approximate these calculations with a certain setting
:
\[
	\bar{\vec{z}} = \ma{D}^{1-\sigma}(\bar{\vec{x}} - \alpha(\ma{K}^{-1}\bar{\vec{x}} - \vec{y}_l))\ma{D}^{\sigma-1}.
\]
Adapting the calculations in Prop. \ref{prop:saferange}, one obtains the safe option as $\alpha=\frac{2\mu}{\mu+4}$.
In the experiments, we examine the classification accuracy of the exact solution and RSF estimates on the benchmark data sets Cora and Citeseer citation network\footnote{These datasets are available at https://linqs.soe.ucsc.edu/data}. Fig.~\ref{fig:ssl} presents the classification accuracy given by $\hat{\vec{x}},\bar{\vec{x}}$ and $\hat{\vec{z}}$.
\begin{figure}
	{\includegraphics[width=8cm,height=4cm]{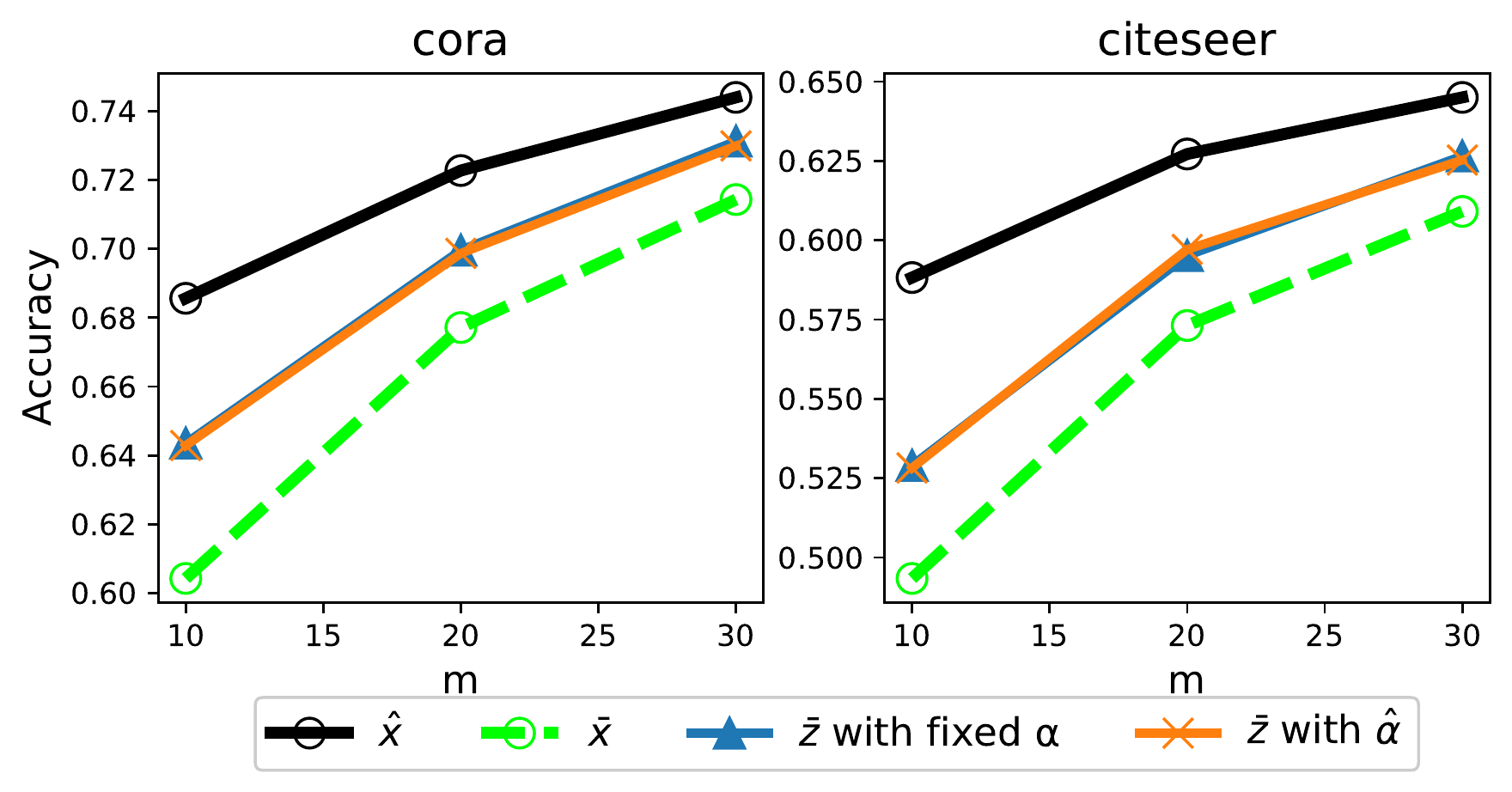}}%
	\caption{Classification performance on the Cora and Citeseer datasets. $m$ denotes the number of labeled vertices per class. In the experiments, we use $\mu=1.0$, $\sigma=0$ and $\alpha =\frac{2\mu}{\mu+4}$. The number of Monte Carlo samples is 50 and the experiments are repeated 100 times for each $m$.  }
\label{fig:ssl}
\end{figure}
We observe that $\bar{\vec{z}}$ with the constant $\alpha$ performs quite similar to $\bar{\vec{z}}$ with $\hat{\alpha}$ while both outperforming $\bar{\vec{x}}$: for this particular example $\frac{2\mu}{\mu+4}$ gives a very good approximation for $\alpha^\star$.

\section{Conclusion}
\label{sec:conc}
The main idea in this work is that, given an unbiased estimator of the solution to a large-scale least-squares problem (such as in the recent lines of works based on DPPs~\cite{derezinski2021determinantal, fanuel2021diversity}), then one (or in fact several) step of classical deterministic gradient descent leaves the estimator unbiased while potentially significantly reducing its variance, for a very small extra cost.

We illustrate this idea on graph Tikhonov regularization via random spanning forests, where we show that the variance of the estimator may easily be divided by several factors if one chooses correctly the gradient descent step $\alpha$. The simplest (and free to compute) choice for $\alpha$ is by far $2q/(q+2d_{max})$. Even though this choice is shown to necessarily decrease the variance, its performance is hindered when the degree distribution of the graph is broad and $d_{max}$ becomes a crude upper bound of the distribution's mode.  

A natural workaround that we will explore in future work is to precondition the gradient step by $\ma{D}^{-1}$ yielding (still unbiased) estimators of the form $\bar{\vec{z}} \coloneqq \bar{\vec{x}} - \alpha \ma{D}^{-1}( \ma{K}^{-1}\bar{\vec{x}} - \vec{y})$.

\vfill\pagebreak

\bibliographystyle{IEEEbib}
\bibliography{strings,refs}

\end{document}